\documentclass[runningheads]{llncs}
\usepackage[T1]{fontenc}   
\usepackage{amsmath}
\usepackage{graphicx}

\newcommand{\junk}[1]{}

\begin{document}
\title{A Note on Solving Problems of
  Substantially Super-linear
  Complexity
  in $N^{o(1)}$ Rounds of the Congested Clique}
%\title{Can Problems of Substantially Super-linear Complexity Be Computed
%in $N^{o(19}$ Rounds of the Congested Clique?}

\author{
  Andrzej Lingas
  \inst{}
}
%\authorrunning{A. Lingas}
\institute{
Department of Computer Science, Lund University, Lund, Sweden. 
\email{Andrzej.Lingas@cs.lth.se}}

\pagestyle{plain}
\maketitle

\begin{abstract}
 We study the possibility of designing $N^{o(1)}$-round protocols for
  problems of substantially super-linear polynomial-time (sequential)
  complexity on the congested clique with about $N^{1/2}$ nodes, where
  $N$ is the input size.  We show that the average time complexity of
  the local computation performed at a clique node (in terms of the
  size of the data received by the node) in such protocols has to be
  substantially larger than the time complexity of the given problem.
\end{abstract}

\section{Introduction}
In the distributed computational model of CONGEST the focus is on the
cost of communication while that of local computations is not
important. The computation is synchronized in rounds.  Initially,
each processing unit has an access to a distinct part of the input.
In each round, each unit can send a message to each other unit that is
connected by a direct communication link. The bit size of such a message
is at most logarithmic in the input size.
In the {\em unicast} variant of
CONGEST the messages sent by a single unit in a single round can be
distinct while in the {\em broadcast} variant they have to be identical.  In
each round, each unit can also perform unlimited local
computation. The main objective is to minimize the number of rounds
necessary to solve the input problem.

The unicast CONGEST variant, where each pair of units is connected by a
direct communication link, forms the {\em congested clique} model. The
processing units in this model are referred to as nodes of the congested clique.

Since the implicit introduction of the congested clique model by
Lotker \emph{et al.} in 2005 \cite{LP-SPP05}, several graph problems
have been studied in this model under the following assumptions.  Each
node represents a distinct vertex of the input graph and knows which
of the other nodes represent the neighbors of its vertex in the input
graph. If the vertices and/or edges in the input graph are weighted
then each node knows the weight of its vertex and/or the weights of
the edges incident to its vertex. In each round, each node can send a
distinct message to each other node as well as receive a message from
each other node. It can also perform an unlimited local
computation. Several graph problems have been shown to admit protocols
using a low number of rounds on the congested clique.  For instance, a
chain of papers on computing a minimum spanning tree on the congested
clique resulted in an $O(1)$-round protocol for this problem described
in \cite{K_21}.

Even non-graph problems have been studied in the congested clique
model, e.g., matrix multiplication \cite{CK19}, sorting and routing \cite{L}, basic
geometric problems for planar point sets \cite{JLLP24}  etc. In all these cases, one
assumes that the basic input items like matrix entries, keys or points
in the plane have logarithmic in the input size bit representation and
that initially each node holds a bunch of roughly $N^{1/2}$ input items,
where $N$ is the total number of the input
items. Specifically, Censor-Hillel {\em et al.}  have shown that
$n \times n$ matrix multiplication admits an $O(n^{0.157})$-round protocol
\cite{CK19} whereas Lenzen has demonstrated that both sorting and
routing admit even $O(1)$-round protocols \cite{L} in this setting.

In this paper, we are concerned with the question if a problem of
substantially super-linear polynomial-time complexity
\footnote{That is $\Omega (N^{\alpha}$-time complexity, where
$\alpha$ is a constant larger than $1$.} can be solved in
$N^{o(1)}$ rounds, where $N$ is the input size,
in the congested clique model with about $N^{\frac 12}$ nodes. Matrix multiplication
(MM) and all-pairs shortest paths (APSP) are examples of such problems
at present. While we do not exclude such a possibility, we give
an evidence that such a solution (if possible) should be highly
non-trivial. Our straightforward argument is explained by the
following example.

Suppose that we have some computational problem on $n\times n$
matrices and the best possible asymptotic running time of a sequential
algorithm for this problem is $n^{\psi},$ where $\psi > 2.$ This
asymptotic bound translates to $N^{\psi/ 2}$, in terms of the input
size $N=n^2.$ Suppose next that we are solving this problem on the
congested clique with $n$ nodes.  Obviously, the $n$ nodes need to
perform jointly $\Omega (n^{\psi})$ units of work, hence on the
average, a node has to perform $\Omega (n^{\psi-1})$ units of
work. Otherwise, there would exist a faster sequential algorithm for
this matrix problem.
Assume that this problem can be solved in $O(1)$ rounds on the
congested clique. Then, each node can receive only data consisting of
$O(n)$ words, i.e., of size $M=O(n).$ Hence, the average time
complexity of the local computation at a node expressed in terms of
$M$ is $\Omega(M^{\psi-1})$ while the time complexity of the fastest
sequential algorithm for this problem is $O(N^{\psi/2}).$ Observe that
$\psi -1> \psi /2$ holds under the assumption $\psi > 2.$ Consequently, the
asymptotic time complexity of the local computation performed at a
node on average is higher than that of the fastest sequential
algorithm for this matrix problem.  While this does not yield a
contradiction with the model of congested clique allowing unlimited
local computations, it suggests that an $O(1)$-round solution of this
matrix problem on the congested clique (if possible) must be highly
non-trivial.  In particular, it hints that one has to perform some
local computations that are asymptotically substantially more complex
(in terms of the size of received data) than just solving some
sub-problems on sub-matrices (of asymptotic time complexity not
exceeding that of the original problem) following from a
decomposition of the problem.

Our results are obtained by a straightforward generalization and
formalization of the example above including:(i) arbitrary problems of
substantially super-linear polynomial-time complexity, (ii)
$N^{o(1)}$-round protocols, and even(iii) $n^{\delta}$-round protocols
for sufficiently small $\delta$ depending on the time complexity of
the problem.

Thus, we show that an $N^{o(1)}$-round protocol for a problem of
substantially super-linear polynomial-time complexity in the congested
clique
model with $N^{\frac 12}$ nodes must involve local computations of
substantially larger asymptotic time complexity than that of the
problem. More exactly, we show that the exponent of the average time
complexity of a local computation in terms of the size of received
data must be larger than that of the time complexity of the problem.
As in the example, our argument is based on the observation that the
total work performed by a protocol for a given problem on the congested
clique cannot be smaller than the sequential time complexity of the
problem.  For this reason, designing such protocols if at all possible
seems highly non-trivial.

\section{The Argument}
In order to simplify the presentation 
we shall assume that the size $N$ of the input to a
problem is a square of a natural number $n.$

We shall formulate our argument in the relaxed model of congest clique
with $N^{\frac 12 }$ nodes (processing units), where each node is
connected by directed communication links with all other nodes, and
initially holds a distinct part of the input of size $N^{\frac 12}.$

We assume that the {\em work} done by a protocol within given $t$
rounds on the congested clique is the total sequential time taken by
the local computations at the nodes of the clique within the $t$
rounds in the unit-cost Random Access Machine model with computer
words of logarithmic length \cite{AHU}. Specifically, we assume that
posting $m$ messages or receiving $m$ messages in a round requires
$\Theta(m)$ work.

For a problem $P$ solvable in sequential time polynomial in the input
size $N,$ let us define the exponent $opt(P)$ of the (sequential) time
complexity of $P$ as the smallest real number not less than $1$, such
that the problem is solvable in $O(N^{opt(P)+\epsilon})$ sequential
time for any positive $\epsilon .$

For a protocol $B$ that solves the problem $P$ using $t_B(N)$ rounds
and performing total work $w_B(N)$ on the congest clique, let us
define the lower bound
$ave(B)$ on
the exponent of the average local time complexity of
$B$ at a node of the clique in terms of the size of received data as
the smallest real such that
$\frac {w_B(N)}{N^{\frac 12 }}=O((t_B(N)N^{\frac
  12})^{ave(B)+\epsilon})$ for any positive $\epsilon .$ Note that in
the latter definition, we assume the upper bound of
$t_B(N)N^{\frac 12}$ on the maximal worst-case size
$N^{\frac 12} + (t_B(N)-1)(N^{\frac 12}-1)$ of the data received by
the local computation at a node within $t_B(N)$ rounds, where the
initial input data has size only $N^{\frac 12}.$ 
\newpage
\begin{theorem}
  Consider a problem $P$ solvable in polynomial time and a protocol $B$
  that solves an input instance of 
  $P$ of size $N$ using $t_B(N)$ rounds on the congested clique.
  If $ave(B)$ is well defined
  then the following inequality holds:
 % If  $2\le t_B(N)$ then
  $$ave(B)\ge \frac {opt(P)-\frac 12
    -o(1)}{\frac 12 +\log_N t_B(N)}.$$
  %otherwise $$ave(B)\ge \frac {opt(P)-\frac 12  -o(1)}{\frac 12}$$
  \end{theorem}
\begin{proof}
  Let $w_B(N)$ be the total work performed by the protocol
  $B$ within the $t_B(N)$ rounds.
  By the definition of $w_B(N),$ the problem $P$ can be solved
  in $O(w_B(N))$ sequential time.
  Simply, we can  simulate the protocol sequentially round after
  round. In each round, we can perform the local computations
  of the consecutive nodes, using $O(1)$ registers for each
  direct link to fetch the messages sent to the current node
  in the previous rounds and to post new messages, respectively.
  Hence, we obtain $w_B(N)=
  \Omega(N^{opt(P)})$ by the definitions of $opt(P)$
  and the fact that $P$ can be solved in $O(w_B(N))$ time.
  %If $2\le t_B(N)$
  Then, by the definition of $ave(B)$,
  we infer that
  $(t_B(N)N^{\frac 12})^{ave(B)}\ge N^{opt(P)-\frac 12
    -o(1)}.$
Consequently, we obtain
$$N^{(\log_N t_B(N) +\frac 12)ave(B)}
\ge  N^{opt(P)-\frac 12 -o(1)}$$
which yields
$$ave(B)\ge \frac {opt(P)-\frac 12-o(1)}{\frac 12 +\log_N t_B(N)}.$$
%Otherwise, in a similar fashion, we obtain
% $$ave(B)\ge \frac {opt(P)-\frac 12 -o(1)}{\frac 12}.$$
\qed
\end{proof}

\begin{corollary}
If $t_B(N)\le
N^{o(1)}$ then 
$ave(B)\ge 2opt(P)-1-o(1).$
  In particular, if $opt(P)>1$  then
 $ave(B)> opt(P).$ 
\end{corollary}

By substituting $N^{\delta}$ for
$t_B(N)$ in Theorem 1 and by straightforward calculations, we obtain also
the following corollary from Theorem 1.

\begin{corollary}
  If  $t_B(n)\le N^{\delta}$
  and $\delta <\frac{ \frac 12 (opt(P)-1)-o(1)}{opt(P)}$
  then $ave(B)> opt(P).$
\end{corollary}
\begin{proof}
  By Theorem 1, it is sufficient to
  solve the inequality stating that the right-hand side of
  the inequality in Theorem 1, with $\delta$ substituted for
  $\log_Nt_B(N)$, is greater than $opt(P)$. The solution follows
  by straightforward calculations with respect to $\delta.$
  \qed
\end{proof}

Our first example
is concerned with
matrix multiplication, say $P=MM.$ We have
$opt(MM)=\omega/2$, where $\omega$ is the exponent of the fast matrix
multiplication \cite{VXZ24}. Corollary 1 implies $ave(B)\ge \omega -1
-o(1)>opt(MM)$ when $t_B(N)=N^{o(1)}$ and $\omega >2.$ The present upper
bound on $\omega $ is $2.371552$ \cite{VXZ24}. This yields $opt(MM)<
1.186.$ Hence, the threshold value of $\delta$ in Corollary 2 has to be
below $\frac {0.093}{1.186}\le 0.0784$ in case of matrix
multiplication. The fastest protocol for $n\times n$ matrix
multiplication on the congested clique with $n$ nodes
uses $O(n^{0.157})$ rounds
\cite{CK19}, i.e., $O(N^{0.0785})$ rounds. Thus, the exponent of
the round complexity of the protocol from \cite{CK19} is very close to the
threshold in Corollary 2 guaranteeing the exponent of
of the polynomial bounding the average time of local
computation (in terms of
the size of received data) to be larger than that $\omega/2$ of matrix
multiplication.
  
Generally, the larger is $opt(P)$ the larger is the gap between
$ave(B)$ and $opt(P),$ e.g., in our second example concerned with the
all-pairs shortest path problem (APSP) in a graph on $n$ vertices.
%that admits a $\tilde{O}(N^{1/6})$-round protocol on the congested clique \cite{CK19}.

Assume $P=APSP$.
No truly subcubic-time\footnote{That is $O(n^{\beta})$-time,
  where $\beta $ is a constant smaller than $3.$}
  sequential algorithm for APSP is known \cite{VW08}.
So the
best known upper bound on $opt(APSP)$ is $1.5.$ For example,
Corollary 1 implies
$ave(B) > 3-1 -o(1)\ge opt(APSP)+0.4$, when $t_B(N)=N^{o(1)}$ and
$opt(APSP) =1.5.$ The threshold value of $\delta$ in Corollary 2 has
to be below $\frac {0.25}{1.5}\le 1/6$ for APSP.  The fastest known protocol
for APSP in an $n$-vertex graph on the congested clique uses
$O(n^{1/3})$ rounds \cite{CK19}, i.e., $O(N^{1/6})$ rounds. Thus,
the exponent of the round complexity of the protocol for APSP
from \cite{CK19} is also slightly over the threshold in Corollary 2.

We may conclude that the fastest known protocols $B$ for MM and APSP
seem to be already quite advanced since their round complexity is
close to thresholds below which $ave(B)>opt(P)$, where
$P\in \{ MM,\ APSP\}.$

\section{Discussion}

Theorem 1 can be interpreted as a trade-off between the number of
rounds and the average time complexity of the local computations (in
terms of the size of the received data) of a protocol solving a problem on the
congested clique with about $N^{1/2}$ nodes.  The smaller is the
number of rounds required by the protocol the larger must be the
average time complexity of the involved local computations.  By
Corollary 1, if the problem has a substantially super-linear
polynomial-time complexity and the protocol uses only $N^{o(1)}$
rounds then the exponent of the polynomial bounding the average time
complexity of local computations (in terms of the size of received
data) has to be larger than that of the polynomial bounding the time
complexity of the input problem (even if all the involved local
computations are necessary and they use optimal algorithms). Such a
scenario is not excluded by the congested clique model as the model
does not restrict the local computational power. Nevertheless, our
results suggest that designing $N^{o(1)}$-round protocols for the
aforementioned problems might be highly non-trivial if at all
possible.  A compression and reuse of data exchanged by the nodes,
and/or performing some identical/similar computations by several
nodes, all in order to save on the communication might be natural but
not necessarily sufficient reasons for the higher asymptotic
complexity of the local computations.

Theorem 1 can be easily extended to include a generalized CONGEST
model with $N^{\alpha}$ nodes (processing units), $0\le \alpha \le 1,$
where each node is connected by directed communication links with
$N^{\beta}$ other nodes, $0\le \beta \le \alpha,$ and initially
holds a distinct part of the input of size $N^{1-\alpha}.$
%We omit this extension since the most interesting consequences
%of the extended theorem hold anyway for $\alpha$ and $\beta$ close to
%$\frac 12$ (cf. Corollaries 1 and 2).

\section*{Acknowledgments}
Thanks go to all who helped to improve the presentation of the note.
%\small
 \bibliographystyle{abbrv}                   
\bibliography{Voronoi2}

\begin{thebibliography}{1}

\bibitem{AHU}
A.~Aho, J.~Hopcroft, and J.~Ullman.
\newblock {\em The Design and Analysis of Computer Algorithms}.
\newblock Addison-Wesley Publishing Company, Reading, 1974.

\bibitem{CK19}
K.~Censor-Hillel, P.~Kaski, J.~H. Korhonen, C.~Lenzen, A.~Paz, and J.~Suomela.
\newblock Algebraic methods in the congested clique.
\newblock {\em Distributed Computing}, 32(6):461--478, 2019.

\bibitem{JLLP24}
J.~Jansson, C.~Levcopoulos, A.~Lingas, and V.~Polishchuk.
\newblock Convex hulls, triangulations, and {Voronoi} diagrams of planar point
  sets on the congested clique.
\newblock arXiv:2305.09987, 2023. Preliminary version in \emph{Proceedings of
  the Thirty-Fifth Canadian Conference on Computational Geometry (CCCG~2023)},
  pages~183--189, 2023.

\bibitem{L}
C.~Lenzen.
\newblock Optimal deterministic routing and sorting on the congested clique.
\newblock In {\em Proceedings of the 2013 ACM Symposium on Principles of
  Distributed Computing (PODC~2013)}, pages 42--50. ACM, 2013.

\bibitem{LP-SPP05}
Z.~Lotker, B.~Patt-Shamir, E.~Pavlov, and D.~Peleg.
\newblock Minimum-weight spanning tree construction in {$O(\log \log n)$}
  communication rounds.
\newblock {\em SIAM Journal on Computing}, 35(1):120--131, 2005.

\bibitem{K_21}
K.~Nowicki.
\newblock A deterministic algorithm for the {MST} problem in constant rounds of
  congested clique.
\newblock In {\em Proceedings of the Fifty-Third Annual ACM SIGACT Symposium on
  Theory of Computing (STOC~2021)}, pages 1154--1165. ACM, 2021.

\bibitem{VW08}
V.~V. Williams and R.~Williams.
\newblock Subcubic equivalences between path, matrix, and triangle problems.
\newblock {\em Journal of the {ACM}}, 65, 2008.

\bibitem{VXZ24}
V.~V. Williams, Y.~Xu, Z.~Xu, and R.~Zhou.
\newblock New bounds for matrix multiplication: from alpha to omega.
\newblock In {\em Proceedings of the Annual ACM-SIAM Symposium on Discrete
  Algorithms (SODA~2024)}. ACM-SIAM, 2024.

\end{thebibliography}

\vfill

\end{document}